\documentclass[sigconf,edbt]{acmart-edbt2018}

\usepackage{graphicx}

\usepackage{caption}
\usepackage{subcaption}
\usepackage{tikz}
\usepackage{mathtools}

\usepackage{graphicx}
\usepackage{hyperref}
\usepackage{textcomp}

\usepackage{algpseudocode}
\usepackage{algorithm}
\usepackage{algorithmicx}

\usepackage{booktabs} 

\setcopyright{rightsretained}

\acmDOI{}

\acmISBN{978-3-89318-078-3}

\acmConference[EDBT 2018]{21st International Conference on Extending Database Technology (EDBT)}{March 26-29, 2018}{Vienna, Austria} 
\acmYear{2018}

\begin{document}
	
\newtheorem{mytheorem}{Theorem}
\newtheorem{mydef}{Definition}

\algrenewcommand\algorithmicindent{0.5em}
\algnewcommand\algorithmicswitch{\textbf{switch}}
\algnewcommand\algorithmiccase{\textbf{case}}
\algnewcommand\algorithmicassert{\texttt{assert}}
\algnewcommand\Assert[1]{\State \algorithmicassert(#1)}
\algdef{SE}[SWITCH]{Switch}{EndSwitch}[1]{\algorithmicswitch\ #1\ \algorithmicdo}{\algorithmicend\ \algorithmicswitch}
\algdef{SE}[CASE]{Case}{EndCase}[1]{\algorithmiccase\ #1}{\algorithmicend\ \algorithmiccase}

\algtext*{EndSwitch}
\algtext*{EndCase}
\algtext*{EndWhile}
\algtext*{EndIf}
\algtext*{EndFor}
\algtext*{EndFunction}

\newif\ifboldnumber
\newcommand{\boldnext}{\global\boldnumbertrue}

\title{Context-Free Path Querying by Matrix Multiplication}

\author{Rustam Azimov}
\affiliation{%
  \institution{Saint Petersburg State University}
  \streetaddress{7/9 Universitetskaya nab.}
  \city{St. Petersburg} 
  \state{Russia} 
  \postcode{199034}
}
\email{rustam.azimov19021995@gmail.com}

\author{Semyon Grigorev}
\affiliation{%
	\institution{Saint Petersburg State University}
	\streetaddress{7/9 Universitetskaya nab.}
	\city{St. Petersburg}
	\state{Russia}  
	\postcode{199034}
}
\email{Semen.Grigorev@jetbrains.com}

\renewcommand{\shortauthors}{}

\begin{abstract}
Graph data models are widely used in many areas, for example, bioinformatics, graph databases. In these areas, it is often required to process queries for large graphs. Some of the most common graph queries are navigational queries. The result of query evaluation is a set of implicit relations between nodes of the graph, i.e. paths in the graph. A natural way to specify these relations is by specifying paths using formal grammars over the alphabet of edge labels. An answer to a context-free path query in this approach is usually a set of triples $(A, m, n)$ such that there is a path from the node $m$ to the node $n$, whose labeling is derived from a non-terminal $A$ of the given context-free grammar. This type of queries is evaluated using the \textit{relational query semantics}. Another example of path query semantics is the \textit{single-path query semantics} which requires presenting a single path from the node $m$ to the node $n$, whose labeling is derived from a non-terminal $A$ for all triples $(A, m, n)$ evaluated using the relational query semantics. There is a number of algorithms for query evaluation which use these semantics but all of them perform poorly on large graphs. One of the most common technique for efficient big data processing is the use of a graphics processing unit (GPU) to perform computations, but these algorithms do not allow to use this technique efficiently. In this paper, we show how the context-free path query evaluation using these query semantics can be reduced to the calculation of the matrix transitive closure. Also, we propose an algorithm for context-free path query evaluation which uses relational query semantics and is based on matrix operations that make it possible to speed up computations by using a GPU.
\end{abstract}


\maketitle

\section{Introduction}
Graph data models are widely used in many areas, for example, bioinformatics~\cite{Bio}, graph databases~\cite{graphDB}. In these areas, it is often required to process queries for large graphs. The most common among graph queries are navigational queries. The result of query evaluation is a set of implicit relations between nodes of the graph, i.e. paths in the graph. A natural way to specify these relations is by specifying paths using formal grammars (regular expressions, context-free grammars) over the alphabet of edge labels. Context-free grammars are actively used in graphs queries because of the limited expressive power of regular expressions. 

The result of context-free path query evaluation is usually a set of triples $(A, m, n)$ such that there is a path from the node $m$ to the node $n$, whose labeling is derived from a non-terminal $A$ of the given context-free grammar. This type of query is evaluated using the \textit{relational query semantics}~\cite{hellingsRelational}. Another example of path query semantics is the \textit{single-path query semantics}~\cite{hellingsPathQuerying} which requires presenting a single path from the node $m$ to the node $n$ whose labeling is derived from a non-terminal $A$ for all triples $(A, m, n)$ evaluated using the relational query semantics. There is a number of algorithms for context-free path query evaluation using these semantics~\cite{GLL, hellingsRelational, RDF, GraphQueryWithEarley}.

Existing algorithms for context-free path query evaluation w.r.t. these semantics demonstrate poor performance when applied to big data. One of the most common technique for efficient big data processing is \textit{GPGPU} (General-Purpose computing on Graphics Processing Units), but these algorithms do not allow to use this technique efficiently. The algorithms for context-free language recognition had a similar problem until Valiant~\cite{valiant} proposed a parsing algorithm which computes a recognition table by computing matrix transitive closure. Thus, the active use of matrix operations (such as matrix multiplication) in the process of a transitive closure computation makes it possible to efficiently apply GPGPU computing techniques~\cite{matricesOnGPGPU}.

We address the problem of creating an algorithm for context-free path query evaluation using the relational and the single-path query semantics which allows us to speed up computations with GPGPU by using the matrix operations.

The main contribution of this paper can be summarized as follows:
\begin{itemize}
	\item We show how the context-free path query evaluation w.r.t. the relational and the single-path query semantics can be reduced to the calculation of matrix transitive closure.
	\item We introduce an algorithm for context-free path query evaluation w.r.t. the relational query semantics which is based on matrix operations that make it possible to speed up computations by means of GPGPU.
	\item We provide a formal proof of correctness of the proposed algorithm.
	\item We show the practical applicability of the proposed algorithm by running different implementations of our algorithm on real-world data.
\end{itemize}

\section{Preliminaries} \label{section_preliminaries}
In this section, we introduce the basic notions used throughout the paper.

Let $\Sigma$ be a finite set of edge labels. Define an \textit{edge-labeled directed graph} as a tuple $D = (V, E)$ with a set of nodes $V$ and a directed edge-relation $E \subseteq V \times \Sigma \times V$.  For a path $\pi$ in a graph $D$, we denote the unique word obtained by concatenating the labels of the edges along the path $\pi$ as $l(\pi)$. Also, we write $n \pi m$ to indicate that a path $\pi$ starts at the node $n \in V$ and ends at the node $m \in V$.

Following Hellings~\cite{hellingsRelational}, we deviate from the usual definition of a context-free grammar in \textit{Chomsky Normal Form}~\cite{chomsky} by not including a special starting non-terminal, which will be specified in the path queries to the graph. Since every context-free grammar can be transformed into an equivalent one in Chomsky Normal Form and checking that an empty string is in the language is trivial it is sufficient to consider only grammars of the following type. A \textit{context-free grammar} is a triple $G = (N, \Sigma, P)$, where $N$ is a finite set of non-terminals, $\Sigma$ is a finite set of terminals, and $P$ is a finite set of productions of the following forms:

\begin{itemize}
    \item $A \rightarrow B C$, for $A,B,C \in N$,
    \item $A \rightarrow x$, for $A \in N$ and $x \in \Sigma$.   
\end{itemize}

Note that we omit the rules of the form $A \rightarrow \varepsilon$, where $\varepsilon$ denotes an empty string. This does not restrict the applicability of our algorithm because only the empty paths $m \pi m$ correspond to an empty string $\varepsilon$.

We use the conventional notation $A \xrightarrow{*} w$ to denote that a string $w \in \Sigma^*$ can be derived from a non-terminal $A$ by some sequence of applications of the production rules from $P$. The \textit{language} of a grammar $G = (N,\Sigma,P)$ with respect to a start non-terminal $S \in N$ is defined by $$L(G_S) = \{w \in \Sigma^*~|~S \xrightarrow{*} w\}.$$

For a given graph $D = (V, E)$ and a context-free grammar $G = (N, \Sigma, P)$, we define \textit{context-free relations} $R_A \subseteq V \times V$, for every $A \in N$, such that $$R_A = \{(n,m)~|~\exists n \pi m~(l(\pi) \in L(G_A))\}.$$

We define a binary operation $(~\cdot~)$ on arbitrary subsets $N_1 , N_2$ of $N$ with respect to a context-free grammar $G = (N, \Sigma, P)$ as $$N_1 \cdot N_2 = \{A~|~\exists B \in N_1, \exists C \in N_2 \text{ such that }(A \rightarrow B C) \in P\}.$$

Using this binary operation as a multiplication of subsets of $N$ and union of sets as an addition, we can define a \textit{matrix multiplication}, $a \times b = c$, where $a$ and $b$ are matrices of a suitable size that have subsets of $N$ as elements, as $$c_{i,j} = \bigcup^{n}_{k=1}{a_{i,k} \cdot b_{k,j}}.$$

According to Valiant~\cite{valiant}, we define the \textit{transitive closure} of a square matrix $a$ as $a^+ = a^{(1)}_+ \cup a^{(2)}_+ \cup \cdots$ where $a^{(1)}_+ = a$ and $$a^{(i)}_+ = \bigcup^{i-1}_{j=1}{a^{(j)}_+ \times a^{(i-j)}_+}, ~i \ge 2.$$

We enumerate the positions in the input string $s$ of Valiant's algorithm from 0 to the length of $s$. Valiant proposes the algorithm for computing this transitive closure only for upper triangular matrices, which is sufficient since for Valiant's algorithm the input is essentially a directed chain and for all possible paths $n \pi m$ in a directed chain $n < m$. In the context-free path querying input graphs can be arbitrary. For this reason, we introduce an algorithm for computing the transitive closure of an arbitrary square matrix.

For the convenience of further reasoning, we introduce another definition of the transitive closure of an arbitrary square matrix $a$ as $a^{cf} = a^{(1)} \cup a^{(2)} \cup \cdots$ where $a^{(1)} = a$ and $$a^{(i)} = a^{(i-1)} \cup (a^{(i-1)} \times a^{(i-1)}), ~i \ge 2.$$

To show the equivalence of these two definitions of transitive closure, we introduce the partial order $\succeq$ on matrices with the fixed size which have subsets of $N$ as elements. For square matrices $a, b$ of the same size, we denote $a \succeq b$ iff $a_{i,j} \supseteq b_{i,j}$, for every $i, j$. For these two definitions of transitive closure, the following lemmas and theorem hold.

\begin{lemma}\label{lemma:cf_geq_valiant}
	Let $G =(N,\Sigma,P)$ be a grammar, let $a$ be a square matrix. Then $a^{(k)} \succeq a^{(k)}_+$ for any $k \geq 1$.
\end{lemma}
\begin{proof}(Proof by Induction)
	
	\textbf{Basis}: The statement of the lemma holds for $k = 1$, since $$a^{(1)} = a^{(1)}_+ = a.$$
	
	\textbf{Inductive step}: Assume that the statement of the lemma holds for any $k \leq (p - 1)$ and show that it also holds for $k = p$ where $p \geq 2$. For any $i \geq 2$ $$a^{(i)} = a^{(i-1)} \cup (a^{(i-1)} \times a^{(i-1)}) \Rightarrow a^{(i)} \succeq a^{(i-1)}.$$ Hence, by the inductive hypothesis, for any $i \leq (p-1)$ $$a^{(p-1)} \succeq a^{(i)} \succeq a^{(i)}_+.$$ Let $1 \leq j \leq (p - 1)$. The following holds $$(a^{(p-1)} \times a^{(p-1)}) \succeq (a^{(j)}_+ \times a^{(p-j)}_+),$$ since $a^{(p-1)} \succeq a^{(j)}_+$ and $a^{(p-1)} \succeq a^{(p-j)}_+$. By the definition, $$a^{(p)}_+ = \bigcup^{p-1}_{j=1}{a^{(j)}_+ \times a^{(p-j)}_+}$$ and from this it follows that $$(a^{(p-1)} \times a^{(p-1)}) \succeq a^{(p)}_+.$$ By the definition, $$a^{(p)} = a^{(p-1)} \cup (a^{(p-1)} \times a^{(p-1)}) \Rightarrow a^{(p)} \succeq (a^{(p-1)} \times a^{(p-1)}) \succeq a^{(p)}_+$$ and this completes the proof of the lemma.
\end{proof}

\begin{lemma}\label{lemma:valiant_geq_cf}
	Let $G =(N,\Sigma,P)$ be a grammar, let $a$ be a square matrix. Then for any $k \geq 1$ there is $j \geq 1$, such that $(\bigcup^{j}_{i=1}{a^{(i)}_+}) \succeq a^{(k)}$.
\end{lemma}
\begin{proof}(Proof by Induction)
	
	\textbf{Basis}: For $k = 1$ there is $j = 1$, such that $$a^{(1)}_+ = a^{(1)} = a.$$ Thus, the statement of the lemma holds for $k = 1$.
	
	\textbf{Inductive step}: Assume that the statement of the lemma holds for any $k \leq (p - 1)$ and show that it also holds for $k = p$ where $p \geq 2$. By the inductive hypothesis, there is $j \geq 1$, such that $$(\bigcup^{j}_{i=1}{a^{(i)}_+}) \succeq a^{(p-1)}.$$ By the definition, $$a^{(2j)}_+ = \bigcup^{2j-1}_{i=1}{a^{(i)}_+ \times a^{(2j-i)}_+}$$ and from this it follows that $$(\bigcup^{2j}_{i=1}{a^{(i)}_+}) \succeq (\bigcup^{j}_{i=1}{a^{(i)}_+}) \times (\bigcup^{j}_{i=1}{a^{(i)}_+}) \succeq (a^{(p-1)} \times a^{(p-1)}).$$ The following holds $$(\bigcup^{2j}_{i=1}{a^{(i)}_+}) \succeq a^{(p)} = a^{(p-1)} \cup (a^{(p-1)} \times a^{(p-1)}),$$ since $$(\bigcup^{2j}_{i=1}{a^{(i)}_+}) \succeq (\bigcup^{j}_{i=1}{a^{(i)}_+}) \succeq a^{(p-1)}$$ and $$(\bigcup^{2j}_{i=1}{a^{(i)}_+}) \succeq (a^{(p-1)} \times a^{(p-1)}).$$ Therefore there is $2j$, such that $$(\bigcup^{2j}_{i=1}{a^{(i)}_+}) \succeq a^{(p)}$$ and this completes the proof of the lemma.	
\end{proof}

\begin{mytheorem}\label{thm:closures}
	Let $G =(N,\Sigma,P)$ be a grammar, let $a$ be a square matrix. Then $a^+ = a^{cf}$.
\end{mytheorem}
\begin{proof}
	
	By the lemma~\ref{lemma:cf_geq_valiant}, for any $k \geq 1$, $a^{(k)} \succeq a^{(k)}_+$. Therefore $$a^{cf} = a^{(1)} \cup a^{(2)} \cup \cdots \succeq a^{(1)}_+ \cup a^{(2)}_+ \cup \cdots = a^+.$$ By the lemma~\ref{lemma:valiant_geq_cf}, for any $k \geq 1$ there is $j \geq 1$, such that $$(\bigcup^{j}_{i=1}{a^{(i)}_+}) \succeq a^{(k)}.$$ Hence $$a^+ = (\bigcup^{\infty}_{i=1}{a^{(i)}_+}) \succeq a^{(k)},$$ for any $k \geq 1$. Therefore $$a^+ \succeq a^{(1)} \cup a^{(2)} \cup \cdots = a^{cf}.$$ Since $a^{cf} \succeq a^+$ and $a^+ \succeq a^{cf}$, $$a^+ = a^{cf}$$ and this completes the proof of the theorem.
\end{proof}

Further, in this paper, we use the transitive closure $a^{cf}$ instead of $a^+$ and, by the theorem~\ref{thm:closures}, an algorithm for computing $a^{cf}$ also computes Valiant's transitive closure $a^+$.

\section{Related works} \label{section_related}
Problems in many areas can be reduced to one of the formal-languages-constrained path problems~\cite{barrett2000formal}. For example, various problems of static code analysis~\cite{bastani2015specification,xu2009scaling} can be formulated in terms of the context-free language reachability~\cite{reps1998program} or in terms of the linear conjunctive language reachability~\cite{zhang2017context}. 

One of the well-known problems in the area of graph database analysis is the language-constrained path querying. For example, the regular language constrained path querying~\cite{reutter2017regular, fan2011adding, abiteboul1997regular, nole2016regular}, and the context-free language constrained path querying.

There are a number of solutions~\cite{hellingsRelational, GraphQueryWithEarley, RDF} for context-free path query evaluation w.r.t. the relational query semantics, which employ such parsing algorithms as CYK~\cite{kasami, younger} or Earley~\cite{Grune}. Other examples of path query semantics are single-path and \textit{all-path query semantics}. The all-path query semantics requires presenting all possible paths from node $m$ to node $n$ whose labeling is derived from a non-terminal $A$ for all triples $(A, m, n)$ evaluated using the relational query semantics. Hellings~\cite{hellingsPathQuerying} presented algorithms for the context-free path query evaluation using the single-path and the all-path query semantics. If a context-free path query w.r.t. the all-path query semantics is evaluated on cyclic graphs, then the query result can be an infinite set of paths. For this reason, in~\cite{hellingsPathQuerying}, annotated grammars are proposed as a possible solution.

In~\cite{GLL}, the algorithm for context-free path query evaluation w.r.t. the all-path query semantics is proposed. This algorithm is based on the generalized top-down parsing algorithm~---~GLL~\cite{scott2010gll}. This solution uses derivation trees for the result representation which is more native for grammar-based analysis. The algorithms in~\cite{GLL, hellingsPathQuerying} for the context-free path query evaluation w.r.t. the all-path query semantics can also be used for query evaluation using the relational and the single-path semantics.

Our work is inspired by Valiant~\cite{valiant}, who proposed an algorithm for general context-free recognition in less than cubic time. This algorithm computes the same parsing table as the CYK algorithm but does this by offloading the most intensive computations into calls to a Boolean matrix multiplication procedure. This approach not only provides an asymptotically more efficient algorithm but it also allows us to effectively apply GPGPU computing techniques. Valiant's algorithm computes the transitive closure $a^+$ of a square upper triangular matrix $a$. Valiant also showed that the matrix multiplication operation $(\times)$ is essentially the same as $|N|^2$ Boolean matrix multiplications, where $|N|$ is the number of non-terminals of the given context-free grammar in Chomsky normal form.

Hellings~\cite{hellingsRelational} presented an algorithm for the context-free path query evaluation using the relational query semantics. According to Hellings, for a given graph $D = (V, E)$ and a grammar $G = (N, \Sigma, P)$ the context-free path query evaluation w.r.t. the relational query semantics reduces to a calculation of the context-free relations $R_A$. Thus, in this paper, we focus on the calculation of these context-free relations. Also, Hellings~\cite{hellingsRelational} presented an algorithm for the context-free path query evaluation using the single-path query semantics which evaluates paths of minimal length for all triples $(A,m,n)$, but also noted that the length of these paths is not necessarily upper bounded. Thus, in this paper, we evaluate an arbitrary path for all triples $(A,m,n)$.

Yannakakis~\cite{transitive-closure} analyzed the reducibility of various path querying problems to the calculation of the transitive closure. He formulated a problem of Valiant's technique generalization to the context-free path query evaluation w.r.t. the relational query semantics. Also, he assumed that this technique cannot be generalized for arbitrary graphs, though it does for acyclic graphs.

Thus, the possibility of reducing the context-free path query evaluation using the relational and the single-path query semantics to the calculation of the transitive closure is an open problem.

\section{Context-free path querying by the calculation of transitive closure}
In this section, we show how the context-free path query evaluation using the relational query semantics can be reduced to the calculation of matrix transitive closure $a^{cf}$, prove the correctness of this reduction, introduce an algorithm for computing the transitive closure $a^{cf}$, and provide a step-by-step demonstration of this algorithm on a small example.

\subsection{Reducing context-free path querying to transitive closure} \label{section_reducing}
In this section, we show how the context-free relations $R_A$ can be calculated by computing the transitive closure $a^{cf}$.

Let $G = (N,\Sigma,P)$ be a grammar and $D = (V, E)$ be a graph. We enumerate the nodes of the graph $D$ from 0 to $(|V| - 1)$. We initialize the elements of the $|V| \times |V|$ matrix $a$ with $\varnothing$. Further, for every $i$ and $j$ we set $$a_{i,j} = \{A_k~|~((i,x,j) \in E) \wedge ((A_k \rightarrow x) \in P)\}.$$ Finally, we compute the transitive closure $$a^{cf} = a^{(1)} \cup a^{(2)} \cup \cdots$$ where $$a^{(i)} = a^{(i-1)} \cup (a^{(i-1)} \times a^{(i-1)}),$$ for $i \ge 2$ and $a^{(1)} = a$. For the transitive closure $a^{cf}$, the following statements hold.

\begin{lemma}\label{lemma:cf}
Let $D = (V,E)$ be a graph, let $G =(N,\Sigma,P)$ be a grammar. Then for any $i, j$ and for any non-terminal $A \in N$, $A \in a^{(k)}_{i,j}$ iff $(i,j) \in R_A$ and $i \pi j$, such that there is a derivation tree of the height $h \leq k$ for the string $l(\pi)$ and a context-free grammar $G_A = (N,\Sigma,P,A)$.
\end{lemma}
\begin{proof}(Proof by Induction)

\textbf{Basis}: Show that the statement of the lemma holds for $k = 1$. For any $i, j$ and for any non-terminal $A \in N$, $A \in a^{(1)}_{i,j}$ iff there is $i \pi j$ that consists of a unique edge $e$ from the node $i$ to the node $j$ and $(A \rightarrow x) \in P$ where $x = l(\pi)$. Therefore $(i,j) \in R_A$ and there is a derivation tree of the height $h = 1$, shown in Figure~\ref{tree1}, for the string $x$ and a context-free grammar $G_A = (N,\Sigma,P,A)$. Thus, it has been shown that the statement of the lemma holds for $k = 1$.

\begin{figure}[h!]
 \centering
 \includegraphics[width=2cm]{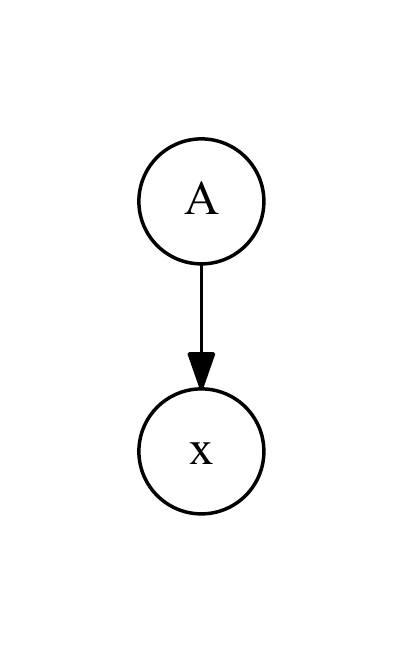}
 \caption{The derivation tree of the height $h = 1$ for the string $x = l(\pi)$.}
 \label{tree1}
\end{figure}

\textbf{Inductive step}: Assume that the statement of the lemma holds for any $k \leq (p - 1)$ and show that it also holds for $k = p$ where $p \geq 2$. For any $i, j$ and for any non-terminal $A \in N$, $$A \in a^{(p)}_{i,j} \text{ iff } A \in a^{(p-1)}_{i,j} \text{ or } A \in (a^{(p-1)} \times a^{(p-1)})_{i,j},$$ since $$a^{(p)} = a^{(p-1)} \cup (a^{(p-1)} \times a^{(p-1)}).$$

Let $A \in a^{(p-1)}_{i,j}$. By the inductive hypothesis, $A \in a^{(p-1)}_{i,j}$ iff $(i,j) \in R_A$ and there exists $i \pi j$, such that there is a derivation tree of the height $h \leq (p-1)$ for the string $l(\pi)$ and a context-free grammar $G_A = (N,\Sigma,P,A)$. The statement of the lemma holds for $k = p$ since the height $h$ of this tree is also less than or equal to $p$.

Let $A \in (a^{(p-1)} \times a^{(p-1)})_{i,j}$. By the definition of the binary operation $(\cdot)$ on arbitrary subsets, $A \in (a^{(p-1)} \times a^{(p-1)})_{i,j}$ iff there are $r$, $B \in a^{(p-1)}_{i,r}$ and $C \in a^{(p-1)}_{r,j}$, such that $(A \rightarrow B C) \in P$. Hence, by the inductive hypothesis, there are $i \pi_1 r$ and $r \pi_2 j$, such that $(i,r) \in R_B$ and $(r,j) \in R_C$, and there are the derivation trees $T_B$ and $T_C$ of heights $h_1 \leq (p-1)$ and $h_2 \leq (p-1)$ for the strings $w_1 = l(\pi_1)$, $w_2 = l(\pi_2)$ and the context-free grammars $G_B$, $G_C$ respectively. Thus, the concatenation of paths $\pi_1$ and $\pi_2$ is $i \pi j$, where $(i,j) \in R_A$ and there is a derivation tree of the height $h = 1 + max(h_1, h_2)$, shown in Figure~\ref{tree2}, for the string $w = l(\pi)$ and a context-free grammar $G_A$.

\begin{figure}[h!]
 \centering
 \includegraphics[width=5cm]{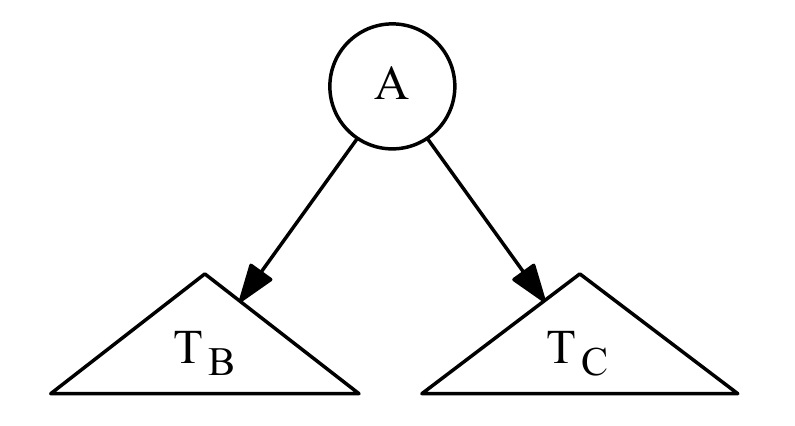}
 \caption{The derivation tree of the height $h = 1 + max(h_1, h_2)$ for the string $w = l(\pi)$, where $T_B$ and $T_C$ are the derivation trees for strings $w_1$ and $w_2$ respectively.}
 \label{tree2}
\end{figure}

The statement of the lemma holds for $k = p$ since the height $h = 1 + max(h_1, h_2) \leq p$. This completes the proof of the lemma.
\end{proof}

\begin{mytheorem}\label{thm:correct}
 Let $D = (V,E)$ be a graph and let $G =(N,\Sigma,P)$ be a grammar. Then for any $i, j$ and for any non-terminal $A \in N$, $A \in a^{cf}_{i,j}$ iff $(i,j) \in R_A$.
\end{mytheorem}
\begin{proof}

Since the matrix $a^{cf} = a^{(1)} \cup a^{(2)} \cup \cdots,$ for any $i, j$ and for any non-terminal $A \in N$, $A \in a^{cf}_{i,j}$ iff there is $k \geq 1$, such that $A \in a^{(k)}_{i,j}$. By the lemma~\ref{lemma:cf}, $A \in a^{(k)}_{i,j}$ iff $(i,j) \in R_A$ and there is $i \pi j$, such that there is a derivation tree of the height $h \leq k$ for the string $l(\pi)$ and a context-free grammar $G_A = (N,\Sigma,P,A)$. This completes the proof of the theorem.
\end{proof}

We can, therefore, determine whether $(i,j) \in R_A$ by asking whether $A \in a^{cf}_{i,j}$. Thus, we show how the context-free relations $R_A$ can be calculated by computing the transitive closure $a^{cf}$ of the matrix $a$.

\subsection{The algorithm} \label{section_algorithm}
In this section, we introduce an algorithm for calculating the transitive closure $a^{cf}$ which was discussed in Section~\ref{section_reducing}.

Let $D = (V, E)$ be the input graph and $G = (N,\Sigma,P)$ be the input grammar.

\begin{algorithm}[H]
\begin{algorithmic}[1]
\caption{Context-free recognizer for graphs}
\label{alg:graphParse}
\Function{contextFreePathQuerying}{D, G}
    
    \State{$n \gets$ the number of nodes in $D$}
    \State{$E \gets$ the directed edge-relation from $D$}
    \State{$P \gets$ the set of production rules in $G$}
    \State{$T \gets$ the matrix $n \times n$ in which each element is $\varnothing$}
    \ForAll{$(i,x,j) \in E$}
    \Comment{Matrix initialization}
        \State{$T_{i,j} \gets T_{i,j} \cup \{A~|~(A \rightarrow x) \in P \}$}
    \EndFor    
    \While{matrix $T$ is changing}
       
        \State{$T \gets T \cup (T \times T)$}
        \Comment{Transitive closure $T^{cf}$ calculation} 
    \EndWhile
\State \Return $T$
\EndFunction
\end{algorithmic}
\end{algorithm}

Note that the matrix initialization in lines \textbf{6-7} of the Algorithm~\ref{alg:graphParse} can handle arbitrary graph $D$. For example, if a graph $D$ contains multiple edges $(i,x_1,j)$ and $(i,x_2,j)$ then both the elements of the set $\{A~|~(A \rightarrow x_1) \in P \}$ and the elements of the set $\{A~|~(A \rightarrow x_2) \in P \}$ will be added to $T_{i,j}$.

We need to show that the Algorithm~\ref{alg:graphParse} terminates in a finite number of steps. Since each element of the matrix $T$ contains no more than $|N|$ non-terminals, the total number of non-terminals in the matrix $T$ does not exceed $|V|^2|N|$. Therefore, the following theorem holds.

\begin{mytheorem}\label{thm:finite}
 Let $D = (V,E)$ be a graph and let $G =(N,\Sigma,P)$ be a grammar. The Algorithm~\ref{alg:graphParse} terminates in a finite number of steps. 
\end{mytheorem}
\begin{proof}
It is sufficient to show, that the operation in the line \textbf{9} of the Algorithm~\ref{alg:graphParse} changes the matrix $T$ only finite number of times. Since this operation can only add non-terminals to some elements of the matrix $T$, but not remove them, it can change the matrix $T$ no more than $|V|^2|N|$ times.
\end{proof}

Denote the number of elementary operations executed by the algorithm of multiplying two $n \times n$ Boolean matrices as $BMM(n)$. According to Valiant, the matrix multiplication operation in the line \textbf{9} of the Algorithm~\ref{alg:graphParse} can be calculated in $O(|N|^2 BMM(|V|))$. Denote the number of elementary operations executed by the matrix union operation of two $n \times n$ Boolean matrices as $BMU(n)$. Similarly, it can be shown that the matrix union operation in the line \textbf{9} of the Algorithm~\ref{alg:graphParse} can be calculated in $O(|N|^2 BMU(n))$. Since the line \textbf{9} of the Algorithm~\ref{alg:graphParse} is executed no more than $|V|^2|N|$ times, the following theorem holds.

\begin{mytheorem}\label{thm:time}
 Let $D = (V,E)$ be a graph and let $G =(N,\Sigma,P)$ be a grammar. The Algorithm~\ref{alg:graphParse} calculates the transitive closure $T^{cf}$ in $O(|V|^2|N|^3(BMM(|V|) + BMU(|V|)))$.
\end{mytheorem}

\subsection{An example} \label{section_example}
In this section, we provide a step-by-step demonstration of the proposed algorithm. For this, we consider the classical \textit{same-generation query}~\cite{FndDB}.

The \textbf{example query} is based on the context-free grammar $G = (N, \Sigma, P)$ where:
\begin{itemize}
    \item The set of non-terminals $N = \{S\}$.
    \item The set of terminals $$\Sigma = \{subClassOf, subClassOf^{-1}, type, type^{-1}\}.$$
    \item The set of production rules $P$ is presented in Figure~\ref{ProductionRulesExampleQuery}.
\end{itemize}

\begin{figure}[h]
   \[
\begin{array}{rccl}
   0: & S & \rightarrow & \text{\textit{subClassOf}}^{-1} \ S \ \text{\textit{subClassOf}} \\ 
   1: & S & \rightarrow & \text{\textit{type}}^{-1} \ S \ \text{\textit{type}} \\ 
   2: & S & \rightarrow & \text{\textit{subClassOf}}^{-1} \ \text{\textit{subClassOf}} \\ 
   3: & S & \rightarrow & \text{\textit{type}}^{-1} \ \text{\textit{type}} \\ 
\end{array}
\]
\caption{Production rules for the example query grammar.}
\label{ProductionRulesExampleQuery}
\end{figure}

Since the proposed algorithm processes only grammars in Chomsky normal form, we first transform the grammar $G$ into an equivalent grammar $G' = (N', \Sigma', P')$ in normal form, where:
\begin{itemize}
    \item The set of non-terminals $N' = \{S, S_1, S_2, S_3, S_4, S_5, S_6\}$.
    \item The set of terminals $$\Sigma' = \{subClassOf, subClassOf^{-1}, type, type^{-1}\}.$$
    \item The set of production rules $P'$ is presented in Figure~\ref{ProductionRulesExampleQueryCNF}.
\end{itemize}

\begin{figure}[h]
   \[
\begin{array}{rccl}
   0: & S & \rightarrow & S_1 \ S_5 \\
   1: & S & \rightarrow & S_3 \ S_6 \\
   2: & S & \rightarrow & S_1 \ S_2 \\
   3: & S & \rightarrow & S_3 \ S_4 \\
   4: & S_5 & \rightarrow & S \ S_2 \\
   5: & S_6 & \rightarrow & S \ S_4 \\
   6: & S_1 & \rightarrow & \text{\textit{subClassOf}}^{-1} \\ 
   7: & S_2 & \rightarrow & \text{\textit{subClassOf}} \\ 
   8: & S_3 & \rightarrow & \text{\textit{type}}^{-1} \\
   9: & S_4 & \rightarrow & \text{\textit{type}} \\ 
\end{array}
\]
\caption{Production rules for the example query grammar in normal form.}
\label{ProductionRulesExampleQueryCNF}
\end{figure}

We run the query on a graph presented in Figure~\ref{ExampleQueryGraph}.

\begin{figure}[h]
\[
    \includegraphics[width=8cm]{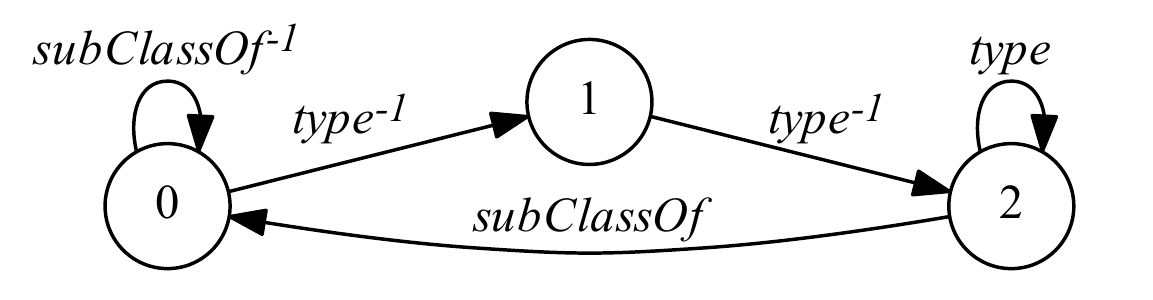}
\]
\caption{An input graph for the example query.}
\label{ExampleQueryGraph}
\end{figure}

We provide a step-by-step demonstration of the work with the given graph $D$ and grammar $G'$ of the Algorithm~\ref{alg:graphParse}. After the matrix initialization in lines \textbf{6-7} of the Algorithm~\ref{alg:graphParse}, we have a matrix $T_0$ presented in Figure~\ref{ExampleQueryInitMatrix}.

\begin{figure}[h]
\[
T_0 = \begin{pmatrix}
    \{S_1\} & \{S_3\} & \varnothing \\ \varnothing & \varnothing & \{S_3\} \\ \{S_2\} & \varnothing & \{S_4\}
\end{pmatrix}
\]
\caption{The initial matrix for the example query.}
\label{ExampleQueryInitMatrix}
\end{figure}

Let $T_i$ be the matrix $T$ obtained after executing the loop in lines \textbf{8-9} of the Algorithm~\ref{alg:graphParse} $i$ times. The calculation of the matrix $T_1$ is shown in Figure~\ref{ExampleQueryFirstIteration}.

\begin{figure}[h]
\[
T_0 \times T_0 = \begin{pmatrix}
    \varnothing & \varnothing & \varnothing \\ \varnothing & \varnothing & \{S\} \\ \varnothing & \varnothing & \varnothing
\end{pmatrix}
\]

\[
T_1 = T_0 \cup (T_0 \times T_0) = \begin{pmatrix}
    \{S_1\} & \{S_3\} & \varnothing \\ \varnothing & \varnothing & \{S_3, S\} \\ \{S_2\} & \varnothing & \{S_4\}
\end{pmatrix}
\]
\caption{The first iteration of computing the transitive closure for the example query.}
\label{ExampleQueryFirstIteration}
\end{figure}

When the algorithm at some iteration finds new paths in the graph $D$, then it adds corresponding nonterminals to the matrix $T$. For example, after the first loop iteration, non-terminal $S$ is added to the matrix $T$. This non-terminal is added to the element with a row index $i = 1$ and a column index $j = 2$. This means that there is $i\pi j$ (a path $\pi$ from the node 1 to the node 2), such that $S \xrightarrow{*} l(\pi)$. For example, such a path consists of two edges with labels $type^{-1}$ and $type$, and thus $S \xrightarrow{*} type^{-1} \ type$.

The calculation of the transitive closure is completed after $k$ iterations when a fixpoint is reached: $T_{k-1} = T_k$. For the example query, $k = 6$ since $T_6 = T_5$. The remaining iterations of computing the transitive closure are presented in Figure~\ref{ExampleQueryFinalIterations}.

\begin{figure}[h]
\[
T_2 = \begin{pmatrix}
    \{S_1\} & \{S_3\} & \varnothing \\ \{S_5\} & \varnothing & \{S_3, S, S_6\} \\ \{S_2\} & \varnothing & \{S_4\}
\end{pmatrix}
\]

\[
T_3 = \begin{pmatrix}
    \{S_1\} & \{S_3\} & \{S\} \\ \{S_5\} & \varnothing & \{S_3, S, S_6\} \\ \{S_2\} & \varnothing & \{S_4\}
\end{pmatrix}
\]

\[
T_4 = \begin{pmatrix}
    \{S_1, S_5\} & \{S_3\} & \{S, S_6\} \\ \{S_5\} & \varnothing & \{S_3, S, S_6\} \\ \{S_2\} & \varnothing & \{S_4\}
\end{pmatrix}
\]

\[
T_5 = \begin{pmatrix}
    \{S_1, S_5, S\} & \{S_3\} & \{S, S_6\} \\ \{S_5\} & \varnothing & \{S_3, S, S_6\} \\ \{S_2\} & \varnothing & \{S_4\}
\end{pmatrix}
\]
\caption{Remaining states of the matrix $T$.}
\label{ExampleQueryFinalIterations}
\end{figure}

Thus, the result of the Algorithm~\ref{alg:graphParse} for the example query is the matrix $T_5 = T_6$. Now, after constructing the transitive closure, we can construct the context-free relations $R_A$. These relations for each non-terminal of the grammar $G'$ are presented in Figure~\ref{ExampleQueryCFRelations}.

\begin{figure}[h]
\begin{eqnarray*}
R_S&=&\{(0,0),(0,2),(1,2)\},\\
R_{S_1}&=&\{(0,0)\},\\
R_{S_2}&=&\{(2,0)\}, \\
R_{S_3}&=&\{(0,1), (1,2)\}, \\
R_{S_4}&=&\{(2,2)\}, \\
R_{S_5}&=&\{(0,0), (1,0)\}, \\
R_{S_6}&=&\{(0,2), (1,2)\}.
\end{eqnarray*}
\caption{Context-free relations for the example query.}
\label{ExampleQueryCFRelations}
\end{figure}

By the context-free relation $R_S$, we can conclude that there are paths in a graph $D$ only from the node 0 to the node 0, from the node 0 to the node 2 or from the node 1 to the node 2, corresponding to the context-free grammar $G_S$. This conclusion is based on the fact that a grammar $G'_S$ is equivalent to the grammar $G_S$ and $L(G_S) = L(G_S')$.

\section{Context-free path querying using single-path semantics}
In this section, we show how the context-free path query evaluation using the single-path query semantics can be reduced to the calculation of matrix transitive closure $a^{cf}$ and prove the correctness of this reduction.

At the first step, we show how the calculation of matrix transitive closure $a^{cf}$ which was discussed in Section~\ref{section_reducing} can be modified to compute the length of some path $i \pi j$ for all $(i,j) \in R_A$, such that $A \xrightarrow{*} l(\pi)$. This is sufficient to solve the problem of context-free path query evaluation using the single-path query semantics since the required path of a fixed length from the node $i$ to the node $j$ can be found by a simple search and checking whether the labels of this path form a string which can be derived from a non-terminal $A$.

Let $G = (N,\Sigma,P)$ be a grammar and $D = (V, E)$ be a graph. We enumerate the nodes of the graph $D$ from 0 to $(|V| - 1)$. We initialize the $|V| \times |V|$ matrix $a$ with $\varnothing$. We associate each non-terminal in matrix $a$ with the corresponding path length. For convenience, each nonterminal $A$ in the $a_{i,j}$ is represented as a pair $(A,k)$ where $k$ is an associated path length. For every $i$ and $j$ we set $$a_{i,j} = \{(A_k,1)~|~((i,x,j) \in E) \wedge ((A_k \rightarrow x) \in P)\}$$ since initially all path lengths are equal to $1$. Finally, we compute the transitive closure $a^{cf}$ and if non-terminal $A$ is added to $a^{(p)}_{i,j}$ by using the production rule $(A \rightarrow B C) \in P$ where $(B,l_B) \in a^{(p-1)}_{i,k}$, $(C,l_C) \in a^{(p-1)}_{k,j}$, then the path length $l_A$ associated with non-terminal $A$ is calculated as $l_A = l_B + l_C$. Therefore $(A, l_A) \in a^{(p)}_{i,j}$. Note that if some non-terminal $A$ with an associated path length $l_1$ is in $a^{(p)}_{i,j}$, then the non-terminal $A$ is not added to the $a^{(k)}_{i,j}$ with an associated path length $l_2$ for all $l_2 \neq l_1$ and $k \geq p$. For the transitive closure $a^{cf}$, the following statements hold.

\begin{lemma}\label{lemma:singlepath}
	Let $D = (V,E)$ be a graph, let $G =(N,\Sigma,P)$ be a grammar. Then for any $i, j$ and for any non-terminal $A \in N$, if $(A,l_A) \in a^{(k)}_{i,j}$, then there is $i \pi j$, such that $A \xrightarrow{*} l(\pi)$ and the length of $\pi$ is equal to $l_A$.
\end{lemma}
\begin{proof}(Proof by Induction)
	
	\textbf{Basis}: Show that the statement of the lemma holds for $k = 1$. For any $i, j$ and for any non-terminal $A \in N$, $(A, l_A) \in a^{(1)}_{i,j}$ iff $l_A = 1$ and there is $i \pi j$ that consists of a unique edge $e$ from the node $i$ to the node $j$ and $(A \rightarrow x) \in P$ where $x = l(\pi)$. Therefore there is $i \pi j$, such that $A \xrightarrow{*} l(\pi)$ and the length of $\pi$ is equal to $l_A$. Thus, it has been shown that the statement of the lemma holds for $k = 1$.
	
	\textbf{Inductive step}: Assume that the statement of the lemma holds for any $k \leq (p - 1)$ and show that it also holds for $k = p$ where $p \geq 2$. For any $i, j$ and for any non-terminal $A \in N$, $(A, l_A) \in a^{(p)}_{i,j}$ iff $(A, l_A) \in a^{(p-1)}_{i,j}$ or $(A, l_A) \in (a^{(p-1)} \times a^{(p-1)})_{i,j}$ since $a^{(p)} = a^{(p-1)} \cup (a^{(p-1)} \times a^{(p-1)}).$
	
	Let $(A, l_A) \in a^{(p-1)}_{i,j}$. By the inductive hypothesis, there is $i \pi j$, such that $A \xrightarrow{*} l(\pi)$ and the length of $\pi$ is equal to $l_A$. Therefore the statement of the lemma holds for $k = p$.
	
	Let $(A, l_A) \in (a^{(p-1)} \times a^{(p-1)})_{i,j}$. By the definition, $(A, l_A) \in (a^{(p-1)} \times a^{(p-1)})_{i,j}$ iff there are $r$, $(B, l_B) \in a^{(p-1)}_{i,r}$ and $(C, l_C) \in a^{(p-1)}_{r,j}$, such that $(A \rightarrow B C) \in P$ and $l_A = l_B + l_C$. Hence, by the inductive hypothesis, there are $i \pi_1 r$ and $r \pi_2 j$, such that $$(B \xrightarrow{*} l(\pi_1)) \wedge(C \xrightarrow{*} l(\pi_2)),$$ where the length of $\pi_1$ is equal to $l_B$ and the length of $\pi_2$ is equal to $l_C$. Thus, the concatenation of paths $\pi_1$ and $\pi_2$ is $i \pi j$, where $A \xrightarrow{*} l(\pi)$ and the length of $\pi$ is equal to $l_A$. Therefore the statement of the lemma holds for $k = p$ and this completes the proof of the lemma.
\end{proof}

\begin{mytheorem}\label{thm:singlepathcorrect}
	Let $D = (V,E)$ be a graph and let $G =(N,\Sigma,P)$ be a grammar. Then for any $i, j$ and for any non-terminal $A \in N$, if $(A, l_A) \in a^{cf}_{i,j}$, then there is $i \pi j$, such that $A \xrightarrow{*} l(\pi)$ and the length of $\pi$ is equal to $l_A$.
\end{mytheorem}
\begin{proof}
	
	Since the matrix $a^{cf} = a^{(1)} \cup a^{(2)} \cup \cdots$, for any $i, j$ and for any non-terminal $A \in N$, if $(A, l_A) \in a^{cf}_{i,j}$, then there is $k \geq 1$, such that $A \in a^{(k)}_{i,j}$. By the lemma~\ref{lemma:singlepath}, if $(A, l_A) \in a^{(k)}_{i,j}$, then there is $i \pi j$, such that $A \xrightarrow{*} l(\pi)$ and the length of $\pi$ is equal to $l_A$. This completes the proof of the theorem.
\end{proof}

By the theorem~\ref{thm:correct}, we can determine whether $(i,j) \in R_A$ by asking whether $(A, l_A) \in a^{cf}_{i,j}$ for some $l_A$. By the theorem~\ref{thm:singlepathcorrect}, there is $i \pi j$, such that $A \xrightarrow{*} l(\pi)$ and the length of $\pi$ is equal to $l_A$. Therefore, we can find such a path $\pi$ of the length $l_A$ from the node $i$ to the node $j$ by a simple search. Thus, we show how the context-free path query evaluation using the single-path query semantics can be reduced to the calculation of matrix transitive closure $a^{cf}$. Note that the time complexity of the algorithm for context-free path querying w.r.t. the single-path semantics no longer depends on the Boolean matrix multiplications since we modify the matrix representation and operations on the matrix elements.

\section{Evaluation}
In this paper, we do not estimate the practical value of the algorithm for the context-free path querying w.r.t. the single-path query semantics, since this algorithm depends significant on the implementation of the path searching. To show the practical applicability of the algorithm for context-free path querying w.r.t. the relational query semantics, we implement this algorithm using a variety of optimizations and apply these implementations to the navigation query problem for a dataset of popular ontologies taken from~\cite{RDF}. We also compare the performance of our implementations with existing analogs from~\cite{GLL,RDF}. These analogs use more complex algorithms, while our algorithm uses only simple matrix operations.

Since our algorithm works with graphs, each RDF file from a dataset was converted to an edge-labeled directed graph as follows. For each triple $(o,p,s)$ from an RDF file, we added edges $(o,p,s)$ and $(s,p^{-1},o)$ to the graph. We also constructed synthetic graphs $g_1$, $g_2$ and $g_3$, simply repeating the existing graphs.

All tests were run on a PC with the following characteristics:
\begin{itemize}
    \item OS: Microsoft Windows 10 Pro
    \item System Type: x64-based PC
    \item CPU: Intel(R) Core(TM) i7-4790 CPU @ 3.60GHz, 3601 Mhz, 4 Core(s), 4 Logical Processor(s)
    \item RAM: 16 GB
    \item GPU: NVIDIA GeForce GTX 1070
    \begin{itemize}
        \item CUDA Cores:		1920 
        \item Core clock:		1556 MHz 
        \item Memory data rate:	8008 MHz
        \item Memory interface:	256-bit 
        \item Memory bandwidth:	256.26 GB/s
        \item Dedicated video memory:	8192 MB GDDR5
    \end{itemize}
\end{itemize}

We denote the implementation of the algorithm from a paper~\cite{GLL} as $GLL$. The algorithm presented in this paper is implemented in F\# programming language~\cite{fsharp} and is available on GitHub\footnote{GitHub repository of the YaccConstructor project: \url{https://github.com/YaccConstructor/YaccConstructor}.}. We denote our implementations of the proposed algorithm as follows:
\begin{itemize}
    \item dGPU (dense GPU) --- an implementation using row-major order for general matrix representation and a GPU for matrix operations calculation. For calculations of matrix operations on a GPU, we use a wrapper for the CUBLAS library from the managedCuda\footnote{GitHub repository of the managedCuda library: \url{https://kunzmi.github.io/managedCuda/}.} library.
    \item sCPU (sparse CPU) --- an implementation using CSR format for sparse matrix representation and a CPU for matrix operations calculation. For sparse matrix representation in CSR format, we use the Math.Net Numerics\footnote{The Math.Net Numerics WebSite: \url{https://numerics.mathdotnet.com/}.} package.
    \item sGPU (sparse GPU) --- an implementation using the CSR format for sparse matrix representation and a GPU for matrix operations calculation. For calculations of the matrix operations on a GPU, where matrices represented in a CSR format, we use a wrapper for the CUSPARSE library from the managedCuda library.
\end{itemize}

We omit $dGPU$ performance on graphs $g_1$, $g_2$ and $g_3$ since a dense matrix representation leads to a significant performance degradation with the graph size growth. 

We evaluate two classical \textit{same-generation queries}~\cite{FndDB} which, for example, are applicable in bioinformatics.

\textbf{Query 1} is based on the grammar $G^1_S$ for retrieving concepts on the same layer, where:
\begin{itemize}
    \item The grammar $G^1 = (N^1, \Sigma^1, P^1)$.
    \item The set of non-terminals $N^1 = \{S\}$.
    \item The set of terminals $$\Sigma^1 = \{subClassOf, subClassOf^{-1}, type, type^{-1}\}.$$
    \item The set of production rules $P^1$ is presented in Figure~\ref{ProductionRulesQuery1}.
\end{itemize}

\begin{figure}[h]
   \[
\begin{array}{rccl}
   0: & S & \rightarrow & \text{\textit{subClassOf}}^{-1} \ S \ \text{\textit{subClassOf}} \\ 
   1: & S & \rightarrow & \text{\textit{type}}^{-1} \ S \ \text{\textit{type}} \\ 
   2: & S & \rightarrow & \text{\textit{subClassOf}}^{-1} \ \text{\textit{subClassOf}} \\ 
   3: & S & \rightarrow & \text{\textit{type}}^{-1} \ \text{\textit{type}} \\ 
\end{array}
\]
\caption{Production rules for the query 1 grammar.}
\label{ProductionRulesQuery1}
\end{figure}

\begin{table*}[ht]
\centering
\caption{Evaluation results for Query 1}
\label{tbl1}

\begin{tabular}{ | c | c | c | c | c | c | c |}
\hline
Ontology & \#triples & \#results & GLL(ms) & dGPU(ms) & sCPU(ms) & sGPU(ms) \\
\hline 
\hline
skos        & 252 & 810 & 10 & 56 & 14 & 12\\
generations & 273 & 2164 & 19 & 62 & 20 & 13\\
travel      & 277 & 2499 & 24 & 69 & 22 & 30\\
univ-bench  & 293 & 2540 & 25 & 81 & 25 & 15\\
atom-primitive & 425 & 15454 & 255 & 190 & 92 & 22\\
biomedical-measure-primitive & 459 & 15156 & 261 & 266 & 113 & 20\\
foaf        & 631 & 4118 & 39 & 154 & 48 & 9\\
people-pets & 640 & 9472 & 89 & 392 & 142 & 32\\
funding     & 1086 & 17634 & 212 & 1410 & 447 & 36\\
wine        & 1839 & 66572 & 819 & 2047 & 797 & 54\\
pizza       & 1980 & 56195 & 697 & 1104 & 430 & 24\\
$g_{1}$     & 8688 & 141072 & 1926 & --- & 26957 & 82\\
$g_{2}$     & 14712 & 532576 & 6246 & --- & 46809 & 185\\
$g_{3}$     & 15840 & 449560 & 7014 & --- & 24967 & 127\\
\hline
\end{tabular}

\end{table*}

\begin{table*}[h]
\centering
\caption{Evaluation results for Query 2}
\label{tbl2}

\begin{tabular}{ | c | c | c | c | c | c | c |}
\hline
Ontology & \#triples & \#results & GLL(ms) & dGPU(ms) & sCPU(ms) & sGPU(ms) \\
\hline 
\hline
skos        & 252 & 1 & 1 & 10 & 2 & 1\\
generations & 273 & 0 & 1 & 9 & 2 & 0\\
travel      & 277 & 63 & 1 & 31 & 7 & 10\\
univ-bench  & 293 & 81 & 11 & 55 & 15 & 9\\
atom-primitive & 425 & 122 & 66 & 36 & 9 & 2\\
biomedical-measure-primitive & 459 & 2871 & 45 & 276 & 91 & 24\\
foaf        & 631 & 10 & 2 & 53 & 14 & 3\\
people-pets & 640 & 37 & 3 & 144 & 38 & 6\\
funding     & 1086 & 1158 & 23 & 1246 & 344 & 27\\
wine        & 1839 & 133 & 8 & 722 & 179 & 6\\
pizza       & 1980 & 1262 & 29 & 943 & 258 & 23\\
$g_{1}$     & 8688 & 9264 & 167 & --- & 21115 & 38\\
$g_{2}$     & 14712 & 1064 & 46 & --- & 10874 & 21\\
$g_{3}$     & 15840 & 10096 & 393 & --- & 15736 & 40\\
\hline
\end{tabular}

\end{table*}

The grammar $G^1$ is transformed into an equivalent grammar in normal form, which is necessary for our algorithm. This transformation is the same as in Section~\ref{section_example}. Let $R_S$ be a context-free relation for a start non-terminal in the transformed grammar.

The result of query 1 evaluation is presented in Table~\ref{tbl1}, where \#triples is a number of triples $(o,p,s)$ in an RDF file, and \#results is a number of pairs $(n,m)$ in the context-free relation $R_S$. We can determine whether $(i,j) \in R_S$ by asking whether $S \in a^{cf}_{i,j}$, where $a^{cf}$ is a transitive closure calculated by the proposed algorithm. All implementations in Table~\ref{tbl1} have the same \#results and demonstrate up to 1000 times better performance as compared to the algorithm presented in~\cite{RDF} for $Q_1$. Our implementation $sGPU$ demonstrates a better performance than $GLL$. We also can conclude that acceleration from the $GPU$ increases with the graph size growth.

\textbf{Query 2} is based on the grammar $G^2_S$ for retrieving concepts on the adjacent layers, where:
\begin{itemize}
    \item The grammar $G^2 = (N^2, \Sigma^2, P^2)$.
    \item The set of non-terminals $N^2 = \{S, B\}$.
    \item The set of terminals $$\Sigma^2 = \{subClassOf, subClassOf^{-1}\}.$$
    \item The set of production rules $P^2$ is presented in Figure~\ref{ProductionRulesQuery2}.
\end{itemize}

\begin{figure}[h]
   \[
\begin{array}{rccl}
   0: & S & \rightarrow & B \ \text{\textit{subClassOf}} \\ 
   1: & S & \rightarrow & \text{\textit{subClassOf}} \\ 
   2: & B & \rightarrow & \text{\textit{subClassOf}}^{-1} \ B \ \text{\textit{subClassOf}} \\ 
   3: & B & \rightarrow & \text{\textit{subClassOf}}^{-1} \ \text{\textit{subClassOf}} \\ 
\end{array}
\]
\caption{Production rules for the query 2 grammar.}
\label{ProductionRulesQuery2}
\end{figure}

The grammar $G^2$ is transformed into an equivalent grammar in normal form. Let $R_S$ be a context-free relation for a start non-terminal in the transformed grammar.

The result of the query 2 evaluation is presented in Table~\ref{tbl2}. All implementations in Table~\ref{tbl2} have the same \#results. On almost all graphs $sGPU$ demonstrates a better performance than $GLL$ implementation and we also can conclude that acceleration from the $GPU$ increases with the graph size growth.

As a result, we conclude that our algorithm can be applied to some real-world problems and it allows us to speed up computations by means of GPGPU.

\section{Conclusion and future work}
In this paper, we have shown how the context-free path query evaluation w.r.t. the relational and the single-path query semantics can be reduced to the calculation of matrix transitive closure. Also, we provided a formal proof of the correctness of the proposed reduction. In addition, we introduced an algorithm for computing this transitive closure, which allows us to efficiently apply GPGPU computing techniques. Finally, we have shown the practical applicability of the proposed algorithm by running different implementations of our algorithm on real-world data.

We can identify several open problems for further research. In this paper, we have considered only two semantics of context-free path querying but there are other important semantics, such as all-path query semantics~\cite{hellingsPathQuerying} which requires presenting all paths for all triples $(A,m,n)$. Context-free path querying implemented with the algorithm~\cite{GLL} can answer the queries in the all-path query semantics by constructing a parse forest. It is possible to construct a parse forest for a linear input by matrix multiplication~\cite{okhotin_cyk}. Whether it is possible to generalize this approach for a graph input is an open question.

In our algorithm, we calculate the matrix transitive closure naively, but there are algorithms for the transitive closure calculation, which are asymptotically more efficient. Therefore, the question is whether it is possible to apply these algorithms for the matrix transitive closure calculation to the problem of context-free path querying.

Also, there are conjunctive~\cite{okhotinConjAndBool} and Boolean grammars~\cite{okhotinBoolean}, which have more expressive power than context-free grammars. Conjunctive language and Boolean path querying problems are undecidable~\cite{hellingsRelational} but our algorithm can be trivially generalized to work on this grammars because parsing with conjunctive and Boolean grammars can be expressed by matrix multiplication~\cite{okhotin_cyk}. It is not clear what a result of our algorithm applied to this grammars would look like. Our hypothesis is that it would produce the upper approximation of a solution. Also, path querying problem w.r.t. the conjunctive grammars can be applied to static code analysis~\cite{zhang2017context}.

From a practical point of view, matrix multiplication in the main loop of the proposed algorithm may be performed on different GPGPU independently. It can help to utilize the power of multi-GPU systems and increase the performance of the context-free path querying.

There is an algorithm~\cite{apspGPU} for transitive closure calculation on directed graphs which generalized to handle graph sizes inherently larger than the DRAM memory available on the GPU. Therefore, the question is whether it is possible to apply this approach to the matrix transitive closure calculation in the problem of context-free path querying.

\section*{Acknowledgments}

We are grateful to Dmitri Boulytchev, Ekaterina Verbitskaia, Marina Polubelova, Dmitrii Kosarev and Dmitry Koznov for their careful reading, pointing out some mistakes, and invaluable suggestions.
This work is supported by grant from JetBrains Research.

\bibliographystyle{ACM-Reference-Format}
\bibliography{graphparsing}

\end{document}